\documentclass[12pt, draftclsnofoot, onecolumn]{IEEEtran}
\usepackage{amssymb}
\usepackage{amsmath}
\usepackage{amsthm}
\usepackage{graphicx}
\usepackage{algorithm}
\usepackage{algorithmic}
\usepackage{cite}
\usepackage{bm}
\usepackage{float}
\usepackage{multirow}
\usepackage{color}
\usepackage{subfigure}
\usepackage{caption}
\usepackage{epstopdf}
\usepackage{hyperref}

\theoremstyle{definition}

\theoremstyle{remark}
\newtheorem{remark}{Remark}

\theoremstyle{plain}
\newtheorem{theorem}{Theorem}
\newtheorem{lemma}{Lemma}
\newtheorem{proposition}{Proposition}
\newtheorem{corollary}{Corollary}

\begin{document}
\title{Downlink and Uplink Intelligent Reflecting Surface Aided Networks: NOMA and OMA}

\author{Yanyu~Cheng,~\IEEEmembership{Member,~IEEE},
        Kwok~Hung~Li,~\IEEEmembership{Senior~Member,~IEEE},
        Yuanwei~Liu,~\IEEEmembership{Senior~Member,~IEEE},
        Kah~Chan~Teh,~\IEEEmembership{Senior~Member,~IEEE},
        and H.~Vincent~Poor,~\IEEEmembership{Life~Fellow,~IEEE}

\thanks{
This article was presented in part at the IEEE Global Communications Conference (GLOBECOM) 2020 \cite{cheng2020IRS}.

Y. Cheng, K. H. Li, and K. C. Teh are with the School of Electrical and Electronic Engineering, Nanyang Technological University, Singapore 639798 (e-mail: ycheng022@e.ntu.edu.sg; ekhli@ntu.edu.sg; ekcteh@ntu.edu.sg).

Y. Liu is with the School of Electronic Engineering and Computer Science, Queen Mary University of London, London E1 4NS, UK (e-mail: yuanwei.liu@qmul.ac.uk).

H. V. Poor is with the Department of Electrical Engineering, Princeton University, Princeton, NJ 08544, USA (e-mail: poor@princeton.edu).
}
}

\maketitle

\begin{abstract}
Intelligent reflecting surfaces (IRSs) are envisioned to provide reconfigurable wireless environments for future communication networks.
In this paper, both downlink and uplink IRS-aided non-orthogonal multiple access (NOMA) and orthogonal multiple access (OMA) networks are studied, in which an IRS is deployed to enhance the coverage by assisting a cell-edge user device (UD) to communicate with the base station (BS).
To characterize system performance, new channel statistics of the BS-IRS-UD link with Nakagami-$m$ fading are investigated.
For each scenario, the closed-form expressions for the outage probability and ergodic rate are derived.
To gain further insight, the diversity order and high signal-to-noise ratio (SNR) slope for each scenario are obtained according to asymptotic approximations in the high-SNR regime.
It is demonstrated that the diversity order is affected by the number of IRS reflecting elements and Nakagami fading parameters, but the high-SNR slope is not related to these parameters.
Simulation results validate our analysis and reveal the superiority of the IRS over the full-duplex decode-and-forward relay.
\end{abstract}


\IEEEpeerreviewmaketitle

\section{Introduction}
\IEEEPARstart{R}{ecently}, intelligent reflecting surfaces (IRSs) have been proposed as a cost-effective solution to enhance the spectral and energy efficiency of future wireless communication networks \cite{cui2019secure}.
Specifically, an IRS consists of a large number of reconfigurable passive elements, and each element can induce a change of amplitude and phase for the incident signal \cite{huang2019reconfigurable,shen2019secrecy}.
By appropriately adjusting amplitude-reflection coefficients and phase-shift variables, it can improve the link quality and enhance the coverage significantly \cite{zhou2020framework,dong2020secure,feng2020deep}.
Compared with conventional communication assisting techniques such as relays, IRSs consume less energy due to passive reflection and can operate in full-duplex (FD) mode without self-interference \cite{qingqing2019towards}.

On the other hand, non-orthogonal multiple access (NOMA) has been proposed as a candidate technique for next-generation wireless communication networks as it can improve the spectral efficiency by allocating multiple user devices (UDs) to a single resource block \cite{liu2018multiple,cheng2019performance}.
It has been demonstrated that NOMA outperforms the conventional orthogonal multiple access (OMA) from the aspects of spectral efficiency, connection density, and user fairness \cite{men2016performance}.
Since allocating many UDs to a single carrier is not practical as it leads to high computational complexity \cite{cheng2020two}, UD pairing is a practical solution that can strike a balance between the performance and the computational complexity \cite{liu2017enhancing,ding2016impact,liu2016cooperative}.

\subsection{Related Work}
\textit{IRS-aided networks:} IRSs have attracted intensive research interest in both academia and industry \cite{shi2019enhanced,ozdogan2019intelligent}.
Accurate channel estimation is crucial for the application of IRSs.
For multi-user multiple-input single-output (MISO) networks assisted by IRSs, a channel estimation protocol was designed in \cite{mishra2019channel}.
To shorten the training sequences, a channel estimation framework was proposed for downlink IRS-aided multiple-input multiple-output (MIMO) networks in \cite{he2019cascaded}.
In \cite{liu2020matrix}, a message-passing based algorithm was designed for uplink IRS-aided MIMO networks.
In \cite{wang2019channel}, a three-phase pilot-based channel estimation scheme was proposed for MISO networks assisted by IRSs to reduce the channel estimation time.
In \cite{bjornson2019intelligent}, the superiority of an IRS was demonstrated as compared with a half-duplex (HD) decode-and-forward (DF) relay.
In \cite{lyu2020spatial}, the performance of a system aided by multiple IRSs was evaluated, and it was revealed that the IRS-aided system outperforms the FD DF relay (FDR)-aided system when the number of IRSs exceeds a certain value.
In \cite{zhang2019analysis}, multiple IRSs were deployed in a single-input single-output (SISO) system, and the phase shifts of IRSs were optimized by minimizing the outage probability (OP).
In \cite{yu2019miso} and \cite{wu2019intelligent}, active and passive beamforming was jointly optimized for IRS-aided MISO networks with different objectives.
The aforementioned works assume that IRSs have continuous phase shifts. On the other hand, discrete phase shifts have attracted considerable attention due to low complexity.
In \cite{xu2019discrete}, the performance of an IRS-aided SISO system with discrete phase shifts was evaluated from the data rate perspective.
In \cite{guo2019weighted}, a MISO IRS-aided system was considered, and the active beamforming at the base station (BS) and the passive beamforming at the IRS were jointly optimized to maximize the weighted sum rate.
In \cite{you2020channel}, the uplink transmission for an IRS-aided system was studied, and the channel estimation and passive beamforming problems were investigated.
In \cite{abeywickrama2020intelligent,wu2019beamforming,ye2020joint}, for IRS-aided networks, active and passive beamforming was jointly optimized with different objectives, such as minimizing the transmit power and minimizing the symbol error rate.

\textit{IRS-aided NOMA networks:} Since both IRS and NOMA are promising techniques for future wireless networks, their combination has been investigated recently \cite{ding2020simple,ding2020impact,zhu2019power,fu2019intelligent,yang2020intelligent,mu2019exploiting,hou2019reconfigurable}.
In \cite{ding2020simple}, an IRS was deployed to improve the coverage by assisting a cell-edge UD in data transmission, where this cell-edge UD is paired with a cell-center UD under the NOMA scheme.
The authors further investigated the impact of random phase shifting and coherent phase shifting for an IRS-aided NOMA system in \cite{ding2020impact}.
In \cite{zhu2019power} and \cite{fu2019intelligent}, the beamforming vectors of the BS and IRS were optimized for an IRS-assisted NOMA system.
The aforementioned works assumed that the BS-IRS-UD channel is non-line-of-sight (NLoS). Since the IRS can be pre-deployed, the path between the BS and the IRS can be line-of-sight (LoS) \cite{yang2020intelligent,mu2019exploiting,hou2019reconfigurable}.
In \cite{yang2020intelligent}, the authors assumed the BS-IRS-UD link to be LoS and optimized the active and passive beamforming vectors for a NOMA system aided by an IRS.
In \cite{mu2019exploiting}, the authors studied an IRS-assisted NOMA system by considering ideal and non-ideal IRSs.
In \cite{hou2019reconfigurable}, the IRS's parameters were designed for a prioritized UD in an IRS-assisted NOMA network.

\subsection{Motivation and Contributions}
Due to the introduction of the IRS, the performance analysis of the IRS-aided system becomes much more difficult than that of the conventional system. The channel statistics of the BS-IRS-UD link are crucial for performance evaluation. After optimizing the parameters of the IRS, it is a challenging task to derive the equivalent channel statistics of the BS-IRS-UD link.
Furthermore, most of the existing works assumed this link to be NLoS, such as \cite{ding2020simple,ding2020impact,zhu2019power,fu2019intelligent}.
Some works assumed that the link is LoS but without deriving exact channel statistics, e.g., \cite{yang2020intelligent,mu2019exploiting,hou2019reconfigurable}.
This motivates us to consider the BS-IRS-UD link to be either LoS or NLoS, and derive new channel statistics for further analysis.
As compared with the aforementioned works, we provide a more comprehensive analysis including both downlink and uplink IRS-aided NOMA and OMA networks.
Besides, since the aim of deploying IRSs is to improve the service coverage and NOMA prefers to pair users with distinctive channel conditions for system performance improvement \cite{ding2016impact,liu2016cooperative}, we consider a NOMA network in which a cell-center UD is paired with a cell-edge UD that cannot communicate with the BS directly and needs assistance from an IRS.
As compared with  \cite{cheng2020IRS}, this paper is more comprehensive, and the contributions can be summarized as follows:
\begin{itemize}
\item We provide a comprehensive study on the OP and ergodic rate (ER) of IRS-aided networks, including both downlink and uplink transmissions for NOMA and OMA.
\item We adopt the Nakagami-$m$ fading model for the BS-IRS-UD link and derive new channel statistics of this link based on the central limit theorem (CLT). Since the CLT-based channel statistics are inaccurate when the channel gain is near $0$, we further derive exact channel statistics for the channel gain near $0$ by utilizing the Laplace transform (LT).
\item We derive the closed-form expressions for the OP and ER for each scenario. To gain further insight, we derive the asymptotic approximations of the OP and ER at a high signal-to-noise ratio (SNR) to obtain the diversity order and high-SNR slope, respectively. We demonstrate that the number of IRS reflecting elements and Nakagami fading parameters affect the diversity order but have no influence on the high-SNR slope.
\item We assume the fixed power allocation for downlink NOMA and demonstrate that power-allocation coefficients affect the OP and ER but have no effect on the diversity order and high-SNR slope.
\item Finally, we compare the performances of IRS-aided and FDR-aided networks. Simulation results reveal the superiority of the IRS over FDR in the high-SNR regime.
\end{itemize}

\subsection{Organization and Notation}
The remainder of this paper is organized as follows: In Section \uppercase\expandafter{\romannumeral2}, the model of the IRS-aided NOMA network is described.
New channel statistics of the BS-IRS-UD link are presented in Section \uppercase\expandafter{\romannumeral3}.
The performance analysis for the downlink is conducted in Section \uppercase\expandafter{\romannumeral4}, followed by the analysis for the uplink in Section \uppercase\expandafter{\romannumeral5}.
Furthermore, numerical and simulation results are presented in Section \uppercase\expandafter{\romannumeral6}.
Finally, our conclusion is drawn in Section \uppercase\expandafter{\romannumeral7}.

In this paper, scalars are denoted by italic letters. Vectors and matrices are denoted by bold-face letters.
For a vector $\mathbf{v}$, $\mathrm{diag}(\mathbf{v})$ denotes a diagonal matrix in which each diagonal element is the corresponding element in $\mathbf{v}$, respectively.
$\mathbf{v}^T$ denotes the transpose of $\mathbf{v}$.
$\mathrm{arg}(\cdot)$ denotes the argument of a complex number.
$\mathbb{C}^{x\times y}$ denotes the space of $x \times y$ complex-valued matrices.
$\mathrm{P_r}(\cdot)$ and $\mathbb{E}(\cdot)$ denote the probability and the expectation, respectively.

\section{Model of IRS-Aided NOMA Networks}
\begin{figure}
\begin{center}
\includegraphics[width=0.5\textwidth]{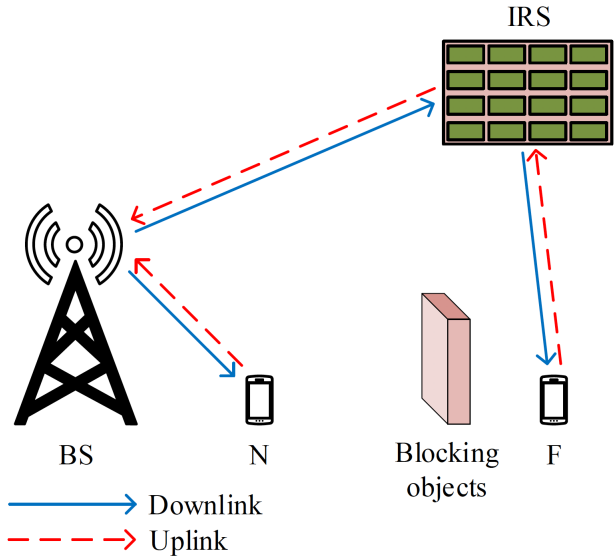}
\end{center}
\caption{Model of the IRS-aided NOMA network.}
\label{fig_system model}
\end{figure}

We consider an IRS-aided NOMA network with a single-antenna BS and two single-antenna UDs, denoted by N and F, respectively, as shown in Fig. \ref{fig_system model}.\footnote{By designing a precoding matrix, we can decompose the MIMO-NOMA into multiple separate SISO-NOMA \cite{ding2017survey}. Hence, in this paper, we only focus on fundamental IRS-aided SISO networks and characterize the corresponding performance. The derived results can be used for investigating IRS-aided MIMO networks in the future.}
More specifically, N is the cell-center UD that can communicate with the BS directly, while F is the cell-edge UD that needs help from an IRS for communications since there is no direct link between the BS and F due to long distance and blocking objects.
The IRS has $K$ reflecting elements, and its reflection-coefficient matrix is denoted by $\mathbf{\Theta}=\mathrm{diag}\left(\beta_1 e^{j\theta_1}, \beta_2 e^{j\theta_2}, \cdots, \beta_K e^{j\theta_K}\right)$ $\left(j=\sqrt{-1}\right)$, where $\beta_k \in [0,1]$ is the amplitude-reflection coefficient and $\theta_k\in [0,2\pi)$ is the phase-shift variable of the $k$th element that can be adjusted by the IRS $\left(k = 1, 2, \cdots, K\right)$.

\subsection{Channel Model}
All channels experience quasi-static flat fading, and the channel state information (CSI) of all channels is assumed to be perfectly known at the BS \cite{ding2020simple,ding2020impact,zhu2019power,fu2019intelligent}.
In particular, several channel estimation schemes for IRS-assisted systems have been proposed to acquire accurate CSI \cite{mishra2019channel,he2019cascaded,liu2020matrix,wang2019channel}.
The link between the BS and N is assumed to be NLoS as it is a link between the BS and the ground UD.
Hence, the small-scale fading between the BS and N follows the Rayleigh fading model and is denoted by $h\thicksim \mathcal{CN}(0,1)$, where $\mathcal{CN}(\cdot,\cdot)$ is the complex Gaussian distribution.
The BS-IRS and IRS-F links can be either LoS or NLoS for different scenarios.
Due to the severe path loss, the signals that are reflected by the IRS twice or more times are ignored.
The small-scale fading vector between the BS and the IRS is denoted by $\mathbf{G} \in \mathbb{C}^{1\times K}$.
The small-scale fading vector between the IRS and F is denoted by $\mathbf{g} \in \mathbb{C}^{K\times 1}$.
Particularly, they are $\mathbf{G}=[G_{1}, G_{2}, \cdots, G_{K}]$ and $\mathbf{g}=[g_{1}, g_{2}, \cdots, g_{K}]^T$, respectively.
All elements in $\mathbf{G}$ and $\mathbf{g}$ follow the Nakagami-$m$ fading model with fading parameters, $m_G$ and $m_g$, respectively.
In particular, it is NLoS for $m_\mathcal{G}=1$ and is LoS for $m_\mathcal{G}>1$ $\left(\mathcal{G} \in \{G, g\}\right)$.

\subsection{Signal Model}
\subsubsection{Downlink}
The BS transmits the signal $x=\sqrt{\alpha_1 P_b}s_{1}^d+\sqrt{\alpha_2 P_b}s_{2}^{d}$, where $P_b$ denotes the transmit power of the BS, $s_{1}^d$ and $s_{2}^{d}$ denote the transmitted signals to N and F, respectively, and $\alpha_1$ and $\alpha_2$ denote the power allocation coefficients for N and F, respectively ($\alpha_1+ \alpha_2=1$).
Note that we assume the fixed power allocation sharing between two UDs and set $\alpha_1 < \alpha_2$ for user fairness \cite{ding2020simple,ding2020impact,hou2019reconfigurable}.\footnote{Optimal power allocation strategies can further enhance the performance of the considered networks, which is beyond the scope of this paper.}
The received signals at N and F are given by
\begin{equation}
\begin{split}
y_N=hd_N^{-\frac{\alpha_h}{2}}x+n_1,
\end{split}
\end{equation}
and
\begin{equation}
\begin{split}
y_F=\left(\mathbf{G} \mathbf{\Theta} \mathbf{g} d_{F1}^{-\frac{\alpha_G}{2}}d_{F2}^{-\frac{\alpha_g}{2}}\right)x+n_2,
\end{split}
\end{equation}
respectively,
where $d_N$ and $d_{F1}$ denote the distances from the BS to N and the IRS, respectively, $d_{F2}$ denotes the distance between the IRS and F, $\alpha_h$, $\alpha_G$, and $\alpha_g$ denote the path loss exponents of BS-N, BS-IRS, and IRS-F links, respectively, and $n_1$ and $n_2$ denote the additive white Gaussian noises (AWGNs) at N and F, respectively, with the same variance $\sigma_n^2$.

At N, the signal of F is detected first, and the corresponding signal-to-interference-plus-noise ratio (SINR) is given by
\begin{equation}
\begin{split}
\mathrm{SINR}_{N,F}^d=\frac{|h|^2d_N^{-\alpha_h} \alpha_2}{|h|^2d_N^{-\alpha_h} \alpha_1+ 1/\rho},
\end{split}
\end{equation}
where $\rho=P_b/\sigma_n^2$ denotes the transmit SNR of the BS.
After implementing the successive interference cancellation (SIC), the signal of N is decoded, and the corresponding SNR is given by
\begin{equation}
\begin{split}
\mathrm{SNR}_N^d=|h|^2d_N^{-\alpha_h} \alpha_1\rho.
\end{split}
\end{equation}
At F, its signal is decoded directly by regarding N's signal as interference, and its SINR is given by
\begin{equation}
\begin{split}
\mathrm{SINR}_{F}^d=\frac{|\mathbf{G} \mathbf{\Theta} \mathbf{g}|^2d_{F1}^{-\alpha_G}d_{F2}^{-\alpha_g}\alpha_2}{|\mathbf{G} \mathbf{\Theta} \mathbf{g}|^2d_{F1}^{-\alpha_G}d_{F2}^{-\alpha_g}\alpha_1+ 1/\rho}.
\end{split}
\end{equation}

\subsubsection{Uplink}
The received signal at the BS is given by
\begin{equation}
\begin{split}
y=h d_N^{-\frac{\alpha_h}{2}} \sqrt{P_u} s_{1}^u + \mathbf{G} \mathbf{\Theta} \mathbf{g} d_{F1}^{-\frac{\alpha_G}{2}}d_{F2}^{-\frac{\alpha_g}{2}} \sqrt{P_u}  s_{2}^{u} + n,
\end{split}
\end{equation}
where $P_u$ denotes the transmit power of each UD, $s_{1}^u$ and $s_{2}^{u} $ denote the transmitted signals from N and F, respectively, and $n$ denotes the AWGN at the BS with variance $\sigma^2$.

At the BS, the signal of N is decoded first by regarding the signal from F as interference, and the corresponding SINR is given by
\begin{equation}\label{eq-UP-SNR1}
\begin{split}
\mathrm{SINR}_N^u=\frac{|h|^2d_N^{-\alpha_h}}{|\mathbf{G} \mathbf{\Theta} \mathbf{g}|^2d_{F1}^{-\alpha_G}d_{F2}^{-\alpha_g}+ 1/\rho'},
\end{split}
\end{equation}
where $\rho'=P_u/\sigma^2$ denotes the transmit SNR of each UD.
After carrying out the SIC, F's signal is detected with the SNR given by
\begin{equation}\label{eq-UP-SNR2}
\begin{split}
\mathrm{SNR}_F^u=|\mathbf{G} \mathbf{\Theta} \mathbf{g}|^2d_{F1}^{-\alpha_G}d_{F2}^{-\alpha_g}\rho'.
\end{split}
\end{equation}

\begin{remark}
The system-design insights can be summarized as follows: 1) By analyzing the performance of the designed system, the suitable number of IRS reflecting elements can be determined to strike a balance between hardware cost and system performance for different scenarios, i.e., downlink (or uplink) NOMA (or OMA) with fixed-rate (or adaptive-rate) transmission. 2) The analytical results contribute to the deployment of IRSs under realistic scenarios, i.e., the BS-IRS and IRS-UD links can be either LoS and NLoS. 3) Based on the impact of power allocation on system performance, we know how to set fixed power-allocation coefficients for different use scenarios of downlink NOMA.
\end{remark}

\section{New Channel Statistics of the BS-IRS-F Link}
In this section, we will present the optimized IRS's parameters and the channel statistics of the BS-IRS-F link.

\subsection{Parameters of the IRS}
For the  BS-IRS-F link, we aim to provide the best channel quality to F by adjusting the parameters of the IRS.
That is to maximize $|\mathbf{G} \mathbf{\Theta} \mathbf{g}| =\left|\sum_{k=1}^{K}\beta_k G_{k}g_{k}e^{j\theta_k}\right|$, where $G_{k}$ and $g_{k}$ are the $k$th element of $\mathbf{G}$ and $\mathbf{g}$, respectively. This can be achieved by intelligently adjusting the phase-shift variable $\theta_k$ for each element, i.e., the phases of all $G_{k}g_{k}e^{j\theta_k}$ are set to be the same.
Therefore, there is not only one solution for $\{\theta_k\}$ $(k = 1, 2, \cdots, K)$, and the generalized solution is given by $\theta_k=\tilde{\theta}-\mathrm{arg}(G_{k}g_{k})$, where $\tilde{\theta}$ is an arbitrary constant ranging in $[0, 2\pi)$.
After adopting the optimal $\{\theta_k\}$, we have
\begin{equation}\label{eq-IRS-Link}
\begin{split}
|\mathbf{G} \mathbf{\Theta} \mathbf{g}|^2 = \beta^2\left(\sum_{k=1}^{K}|G_{k}||g_{k}|\right)^2,
\end{split}
\end{equation}
where we assume that $\beta_k=\beta$, $\forall k$ without loss of generality.

\subsection{New Channel Statistics}
\begin{figure}
\begin{center}
\includegraphics[width=0.55\textwidth]{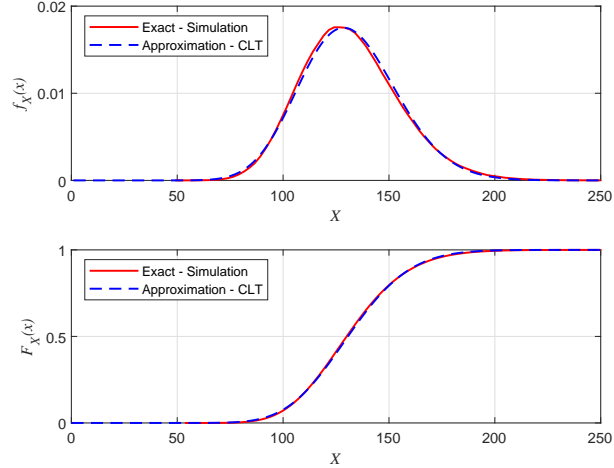}
\end{center}
\caption{The PDF and CDF of $X$ derived by the CLT when $K=30$, $m_G=3$, and $m_g=2$.}
\label{fig_Channel_Statistics}
\end{figure}

New channel statistics can be obtained by using the CLT as shown in the following lemma, which has been verified by Monte Carlo simulations as shown in Fig. \ref{fig_Channel_Statistics}.
\begin{lemma}\label{Lemm-Channel-CLT}
Denote that $X=\frac{\left(\sum_{k=1}^{K}|G_{k}||g_{k}|\right)^2}{K(1-\xi)}$, where $\xi=\frac{1}{m_Gm_g}\left(\frac{\Gamma\left(m_G+\frac{1}{2}\right)}{\Gamma(m_G)}\right)^2\left(\frac{\Gamma\left(m_g+\frac{1}{2}\right)}{\Gamma(m_g)}\right)^2$.
When the number of reflecting elements $K$ is large, $X$ tends to follow a noncentral chi-square distribution as $X \dot\thicksim \chi_1^{'2}(\lambda)$, where $\lambda=\frac{K\xi}{1-\xi}$.
Its probability density function (PDF) and cumulative distribution function (CDF) are given by
\begin{equation}
\begin{split}
f_X(x)=\frac{\lambda^{\frac{1}{4}}}{2}e^{-\frac{x+\lambda}{2}}x^{-\frac{1}{4}}I_{-\frac{1}{2}}\left(\sqrt{\lambda x}\right)
=e^{-\frac{x+\lambda}{2}}\sum_{i=0}^{\infty}\frac{\lambda^i x^{i-\frac{1}{2}}}{i!2^{2i+\frac{1}{2}}\Gamma\left(i+\frac{1}{2}\right)},
\end{split}
\end{equation}
and
\begin{equation}
\begin{split}
F_X(x)=1-Q_{\frac{1}{2}}\left(\sqrt{\lambda},\sqrt{x}\right)
=e^{-\frac{\lambda}{2}}\sum_{i=0}^{\infty}\frac{\lambda^i\gamma\left(i+\frac{1}{2},\frac{x}{2}\right)}{i!2^i\Gamma\left(i+\frac{1}{2}\right)},
\end{split}
\end{equation}
for $x\ge 0$, respectively, where $I_{v}\left(\cdot\right)$ is the modified Bessel function of the first kind, $Q_{v}\left(\cdot, \cdot\right)$ is the Marcum Q-function, $\Gamma(\cdot)$ is the gamma function, and $\gamma(\cdot,\cdot)$ is the lower incomplete gamma function.
\end{lemma}

\begin{proof}
See Appendix \ref{Appen-Lemm-Channel-CLT}.
\end{proof}

\begin{remark}\label{Rema-1}
The CLT-based PDF, $f_X(x)$, is not accurate for $x\rightarrow0^{+}$, which will result in that the derived OP for downlink transmission is inaccurate in the high-SNR regime \cite{ding2020impact}.
\end{remark}

To address the issue stated in Remark \ref{Rema-1}, we further derive exact channel statistics for the channel gain near $0$ without using the CLT as shown in the following lemma.
\begin{lemma}\label{Lemm-Channel-Appro}
Denote that $Z=\sum_{k=1}^{K}|G_{k}||g_{k}|$. When $m_G\neq m_g$, the PDF and CDF of $Z$ for $z\rightarrow0^+$ are given by
\begin{equation}\label{eq-Lemm-1}
\begin{split}
f_Z^{0^+}(z)=\frac{\tilde{m}^K}{\Gamma(2m_sK)}z^{2m_sK-1}e^{-2\sqrt{m_sm_l}z},
\end{split}
\end{equation}
and
\begin{equation}\label{eq-Lemm-2}
\begin{split}
F_Z^{0^+}(z)=\frac{\tilde{m}^K\left(4m_sm_l\right)^{-m_sK}}{\Gamma(2m_sK)}\gamma\left(2m_sK,2\sqrt{m_sm_l}z\right),
\end{split}
\end{equation}
for $z\ge 0$, respectively,
where $m_s=\min\{m_G,m_g\}$, $m_l=\max\{m_G,m_g\}$, and\\ $\tilde{m}=\frac{\sqrt{\pi}4^{m_s-m_l+1}(m_sm_l)^{m_s}\Gamma(2m_s)\Gamma(2m_l-2m_s)}{\Gamma(m_s)\Gamma(m_l)\Gamma\left(m_l-m_s+\frac{1}{2}\right)}$.
\end{lemma}

\begin{proof}
See Appendix \ref{Appen-Lemm-Channel-Appro}.
\end{proof}

\section{Performance Analysis for Downlink Transmission}
In this section, we will investigate the performance of downlink NOMA and OMA networks.
In particular, for each scenario, the OP for the fixed-rate transmission and the ER for the adaptive-rate transmission are derived.

\subsection{NOMA}
\subsubsection{Outage Probability}
For the fixed-rate transmission system, OP is a widely used metric to measure system performance.
The OPs of N and F are given by
\begin{equation}
\begin{split}
\mathbb{P}_N^d=1-\mathrm{P_r}\left(\mathrm{SINR}_{N,F}^d \ge \tilde{\gamma}_F, \mathrm{SNR}_N^d \ge \tilde{\gamma}_N\right),
\end{split}
\end{equation}
and
\begin{equation}
\begin{split}
\mathbb{P}_F^d=\mathrm{P_r}\left(\mathrm{SINR}_F^d < \tilde{\gamma}_F\right),
\end{split}
\end{equation}
respectively, where $\tilde{\gamma}_N=2^{\tilde{R}_N}-1$ and $\tilde{\gamma}_F=2^{\tilde{R}_F}-1$ with $\tilde{R}_N$ and $\tilde{R}_F$ being the target rates of N and F, respectively.

Based on the new channel statistics in Lemma \ref{Lemm-Channel-CLT}, the OPs of N and F are derived in the following theorem.
\begin{theorem}\label{Theo-Downlink-OP}
In the considered IRS-aided NOMA network, the OPs of N and F for the downlink are given by
\begin{equation}\label{eq-NOMA-OP1}
\begin{split}
\mathbb{P}_N^d=1-e^{-\tilde{\rho}_{m}},
\end{split}
\end{equation}
and
\begin{equation}\label{eq-NOMA-OP2}
\begin{split}
\mathbb{P}_F^d\approx e^{-\frac{\lambda}{2}}\sum_{i=0}^{\infty}\frac{\lambda^i \gamma\left(i+\frac{1}{2},\frac{\tilde{\rho}}{2}\right)}{i!2^{i}\Gamma\left(i+\frac{1}{2}\right)},
\end{split}
\end{equation}
respectively, where $\tilde{\rho}_{m}=\max\left\{\frac{\tilde{\gamma}_F}{{\left(\alpha_2-\alpha _1\tilde{\gamma}_F\right)a\rho}}, \frac{\tilde{\gamma}_N}{{a\alpha_1\rho}}\right\}$ and $\tilde{\rho}=\frac{\tilde{\gamma}_F}{{(\alpha_2-\alpha_1\tilde{\gamma}_F)b\rho}}$ with $a=d_N^{-\alpha_h}$ and $b=K\beta^2(1-\xi)d_{F1}^{-\alpha_G}d_{F2}^{-\alpha_g}$.
Note that when we set power allocation coefficients, we need to ensure that $\alpha_2-\alpha_1\tilde{\gamma}_F>0$.
\end{theorem}

\begin{proof}
First, we denote that $Y=|h|^2$, and its PDF and CDF are given by $f_Y(y)=e^{-y}$ for $y\ge 0$ and $F_Y(y)=1-e^{-y}$ for $y\ge 0$, respectively.
Then, $\mathbb{P}_N^d$ can be derived as
\begin{equation}
\begin{split}
\mathbb{P}_N^d=1- \mathrm{P_r}\left(\frac{a\alpha_2 Y}{a\alpha_1 Y + 1/\rho} \ge \tilde{\gamma}_F, a\alpha_1 \rho Y \ge \tilde{\gamma}_N\right)
=F_Y(\tilde{\rho}_{m}).
\end{split}
\end{equation}
On the other hand, $\mathbb{P}_F^d$ can be derived as
\begin{equation}
\begin{split}
\mathbb{P}_F^d=\mathrm{P_r}\left(\frac{b\alpha_2 X}{b\alpha_1 X + 1/\rho} < \tilde{\gamma}_F\right)
\approx F_X(\tilde{\rho}).
\end{split}
\end{equation}
This completes the proof.
\end{proof}

\begin{remark}\label{Rema-3}
According to Remark \ref{Rema-1}, the OP of F is accurate in the low-SNR regime but is inaccurate in the high-SNR regime.
Because the OP of F for $\rho\rightarrow\infty$ is derived by using $\int_{0}^{0^{+}}f_X(x)dx$.
\end{remark}

To solve the aforementioned problem, we can use Lemma \ref{Lemm-Channel-Appro} to derive the high-SNR approximation of F's OP in the following proposition.
\begin{proposition}\label{Prop-Downlink-OP}
In the high-SNR regime, $\mathbb{P}_N^{d}$ can be approximated as
\begin{equation}\label{eq-NOMA-OP3}
\begin{split}
\mathbb{P}_N^{d,\infty}=\tilde{\rho}_{m}.
\end{split}
\end{equation}
When $m_G\neq m_g$, the high-SNR approximation of $\mathbb{P}_F^{d}$ is given by
\begin{equation}\label{eq-NOMA-OP4}
\begin{split}
\mathbb{P}_F^{d,\infty}=\frac{\tilde{m}^K\tilde{c}_1^{2m_sK}}{\Gamma(2m_sK+1)}\rho^{-m_sK},
\end{split}
\end{equation}
where $\tilde{c}_1=\sqrt{\frac{\tilde{\gamma}_F}{c(\alpha_2-\alpha_1\tilde{\gamma}_F)}}$ with $c=\beta^2 d_{F1}^{-\alpha_G}d_{F2}^{-\alpha_g}$.
\end{proposition}

\begin{proof}
By expanding the exponential function in \eqref{eq-NOMA-OP1} and extracting the leading-order term, we can obtain \eqref{eq-NOMA-OP3}.
On the other hand, by expanding the lower incomplete gamma function in \eqref{eq-Lemm-2}, we have
\begin{equation}
\begin{split}
F_Z^{0^+}(z)&=\tilde{m}^K\sum_{l=0}^{\infty}\frac{(2\sqrt{m_sm_l})^{l}z^{l+2m_sK}}{\Gamma(l+2m_sK+1)}e^{-2\sqrt{m_sm_l}z}.
\end{split}
\end{equation}
Then, by using the expansion of the exponential function, we have
\begin{equation}\label{eq-4-2}
\begin{split}
F_Z^{0^+}(z)=\tilde{m}^K\sum_{l=0}^{\infty}\frac{(2\sqrt{m_sm_l})^{l}z^{l+2m_sK}}{\Gamma(l+2m_sK+1)}
\sum_{i=0}^{\infty}\frac{\left(-2\sqrt{m_sm_l}\right)^iz^i}{i!}.
\end{split}
\end{equation}
Meanwhile, the OP of F for $\rho\rightarrow\infty$ can be derived as
\begin{equation}\label{eq-4-3}
\begin{split}
\mathbb{P}_F^{d,\infty}=\mathrm{P_r}\left(\frac{c\alpha_2 Z^2}{c\alpha_1 Z^2 + 1/\rho} < \tilde{\gamma}_F\right)=F_Z^{0^+}\left(\tilde{c}_1\rho^{-\frac{1}{2}}\right).
\end{split}
\end{equation}
Finally, by substituting $z=\tilde{c}_1\rho^{-\frac{1}{2}}$ into \eqref{eq-4-2} and extracting the leading-order term, we can approximate $\mathbb{P}_F^{d,\infty}$ as \eqref{eq-NOMA-OP4}.
This completes the proof.
\end{proof}

\begin{corollary}\label{Coro-Downlink-OP}
In the considered IRS-aided NOMA network, the diversity orders of N and F for the downlink are given by $\mathcal{D}_N^d=1$ and $\mathcal{D}_F^d=m_sK$, respectively.
\end{corollary}
\begin{proof}
Based on Proposition \ref{Prop-Downlink-OP}, we can obtain that the diversity order of F is $m_sK$ when $m_G\neq m_g$. For the case of $m_G=m_g$, it is also $m_sK$ according to the limit. This completes the proof.
\end{proof}

\begin{remark}
The diversity order of F for the downlink NOMA is affected by the number of reflecting elements and Nakagami fading parameters.
\end{remark}

\subsubsection{Ergodic Rate}
For the adaptive-rate transmission system, ER is a commonly used metric to measure system performance.
The ERs of N and F are given by
\begin{equation}
\begin{split}
R_N^d&=\mathbb{E}\left(\log_2\left(1+\mathrm{SNR}_N^d\right)\right),
\end{split}
\end{equation}
and
\begin{equation}\label{eq-NOMA-ER}
\begin{split}
R_F^d&=\mathbb{E}\left(\log_2\left(1+\mathrm{SINR}_F^d\right)\right),
\end{split}
\end{equation}
respectively.
\begin{remark}
For the case of adaptive-rate transmission, the SINR at F is $\min\left\{\mathrm{SINR}_{N,F}^d,\mathrm{SINR}_{F}^d\right\}$.
However, in the considered system, F has a much more severe path loss than N due to a much larger distance to the BS. Therefore, we have $\min\left\{\mathrm{SINR}_{N,F}^d,\mathrm{SINR}_{F}^d\right\}\approx \mathrm{SINR}_{F}^d$.
\end{remark}
After some mathematical manipulations, the ERs of N and F are obtained in the following theorem.
\begin{theorem}\label{Theo-Downlink-ER}
In the considered IRS-aided NOMA network, the ERs of N and F for the downlink are, respectively, given by
\begin{equation}\label{eq-NOMA-ER1}
\begin{split}
R_N^d=-\frac{e^{\frac{1}{a\alpha_1 \rho}}}{\ln(2)}\mathrm{Ei}\left(-\frac{1}{a\alpha_1 \rho}\right),
\end{split}
\end{equation}
and
\begin{equation}\label{eq-NOMA-ER2}
\begin{split}
R_F^d\approx \log_2(1+\tilde{\alpha})-\frac{e^{-\frac{\lambda}{2}}}{\ln(2)}\sum_{i=0}^{\infty}
\frac{\lambda^i}{i!2^i\Gamma\left(i+\frac{1}{2}\right)}\sum_{l=1}^{u_1}\omega_{1,l}\mathcal{J}_1(t_l),
\end{split}
\end{equation}
where $\mathrm{Ei}(\cdot)$ is the exponential integral, $\tilde{\alpha}=\frac{\alpha_2}{\alpha_1}$, $u_1$ is the number of nodes for the Chebyshev-Gauss quadrature and determines the approximation precision, $\omega_{1,l}=\frac{\pi}{u_1}$ is the weight, $t_l=cos\left(\frac{2l-1}{2u_1}\pi\right)$, and
\begin{equation}
\begin{split}
\mathcal{J}_1(t)=\gamma\left(i+\frac{1}{2},\frac{\tilde{\alpha}(1+t)}{4b\rho\alpha_2-2b\rho\alpha_1\tilde{\alpha}(1+t)}\right)\frac{\sqrt{1-t^2}}{1+\frac{2}{\tilde{\alpha}}+t}.
\end{split}
\end{equation}
\end{theorem}

\begin{proof}
See Appendix \ref{Appen-Theo-Downlink-ER}.
\end{proof}

\begin{remark}
The ER of F is accurate as the calculation is not related to $\int_{0}^{0^{+}}g(x)f_X(x)dx$.
Because of the principle of downlink NOMA, the ER of N increases for the increase of the SNR, while the ER of F approaches a ceiling when $\rho\rightarrow\infty$.
\end{remark}

To provide insight into the performance, the high-SNR slope is considered, which is defined as $S=\lim\limits_{\rho\rightarrow\infty}\frac{R(\rho)}{\log_2(\rho)}$ \cite{lozano2005high}.
To obtain it, the asymptotic expression for N's ER and the ceiling for F's ER are derived in the following proposition.
\begin{proposition}\label{Prop-Downlink-ER}
In the high-SNR regime, $R_N^d$ can be approximated as
\begin{equation}\label{eq-NOMA-ER3}
\begin{split}
R_N^{d,\infty}=\log_2(a\alpha_1 \rho)-\frac{E_c}{\ln(2)},
\end{split}
\end{equation}
where $E_c$ denotes the Euler constant.
On the other hand, the ceiling for $R_F^d$ in the high-SNR regime is given by
\begin{equation}\label{eq-NOMA-ER4}
\begin{split}
R_F^{d,\infty}=\log_2(1+\tilde{\alpha}).
\end{split}
\end{equation}
\end{proposition}

\begin{proof}
By using $\lim\limits_{x\rightarrow 0}e^{x} = 1$ and
$\mathrm{Ei}(-x)\approx\ln(x)+E_c$ for $x\rightarrow 0$ \cite[\textrm{eq}. (8.214.2)]{gradshteyn2007}, we can approximate \eqref{eq-NOMA-ER1} as \eqref{eq-NOMA-ER3} when $\rho\rightarrow\infty$.
On the other hand, since we have $\lim\limits_{\rho\rightarrow\infty}\frac{b\alpha_2 X}{b\alpha_1 X + 1/\rho}=\tilde{\alpha}$, we can approximate \eqref{eq-NOMA-ER} as \eqref{eq-NOMA-ER4} when $\rho\rightarrow\infty$. This completes the proof.
\end{proof}

\begin{corollary}\label{Coro-Downlink-ER}
In the considered IRS-aided NOMA network, the high-SNR slopes of N and F for the downlink are given by $\mathcal{S}_N^d=1$ and $\mathcal{S}_F^d=0$, respectively.
\end{corollary}

\begin{proof}
We have $\mathcal{S}_N^d=\frac{dR_N^{d,\infty}}{d\log_2(\rho)}=\ln(2)\rho\frac{dR_N^{d,\infty}}{d\rho}=1$.
This completes the proof.
\end{proof}

\begin{remark}
The high-SNR slope of F for the downlink NOMA is not related to the number of reflecting elements and Nakagami fading parameters.
\end{remark}

\subsection{OMA}\label{Subsec-OMA}
OMA can be regarded as a special case of NOMA.
The data rates of N and F  under the OMA scheme are given by $\frac{1}{2}\log_2\left(1+|h|^2d_N^{-\alpha_h}\rho\right)$
and $\frac{1}{2}\log_2\left(1+\beta^2\left(\sum_{k=1}^{K}|G_{k}||g_{k}|\right)^2d_{F1}^{-\alpha_G}d_{F2}^{-\alpha_g}\rho\right)$,
respectively. Here, we assume that each UD shares half of the resource block for a fair comparison.

\subsubsection{Outage Probability}
The OPs of N and F for the downlink are presented in the following theorem.
\begin{theorem}\label{Theo-OMA-OP}
When N and F are under the downlink OMA scheme, their OPs are given by
\begin{equation}\label{eq-OMA-OP1}
\begin{split}
\mathbb{P}_N^{do}=1-e^{-\frac{\tilde{\gamma}_N^o}{a\rho}},
\end{split}
\end{equation}
and
\begin{equation}\label{eq-OMA-OP2}
\begin{split}
\mathbb{P}_F^{do}\approx e^{-\frac{\lambda}{2}}\sum_{i=0}^{\infty}
\frac{\lambda^i\gamma\left(i+\frac{1}{2},\frac{\tilde{\gamma}_F^o}{2b\rho}\right)}{i!2^{i}\Gamma\left(i+\frac{1}{2}\right)},
\end{split}
\end{equation}
respectively,
where $\tilde{\gamma}_N^o=2^{2\tilde{R}_N}-1$ and $\tilde{\gamma}_F^o=2^{2\tilde{R}_F}-1$.
\end{theorem}

\begin{proof}
$\mathbb{P}_N^{do}$ and $\mathbb{P}_F^{do}$ can be derived as follows:
\begin{equation}
\begin{split}
\mathbb{P}_N^{do}=\mathrm{P_r}\left(a\rho Y<\tilde{\gamma}_N^o\right)=F_Y\left(\frac{\tilde{\gamma}_N^o}{a\rho}\right),
\end{split}
\end{equation}
and
\begin{equation}
\begin{split}
\mathbb{P}_F^{do}=\mathrm{P_r}\left(b\rho X<\tilde{\gamma}_F^o\right)\approx F_X\left(\frac{\tilde{\gamma}_F^o}{b\rho}\right).
\end{split}
\end{equation}
This completes the proof.
\end{proof}

\begin{remark}\label{Rema-7}
According to Remark \ref{Rema-1}, $\mathbb{P}_F^{do}$ is accurate in the low-SNR regime but is inaccurate in the high-SNR regime.
Because $\mathbb{P}_F^{do}$ for $\rho\rightarrow\infty$ is derived by using $\int_{0}^{0^{+}}f_X(x)dx$.
\end{remark}

To address the aforementioned problem, we use Lemma \ref{Lemm-Channel-Appro} and obtain the high-SNR approximation of $\mathbb{P}_F^{do}$ in the following proposition.
\begin{proposition}\label{Prop-OMA-OP}
In the high-SNR regime, $\mathbb{P}_N^{do}$ can be approximated as
\begin{equation}\label{eq-OMA-OP3}
\begin{split}
\mathbb{P}_N^{do,\infty}=\frac{\tilde{\gamma}_N^o}{a}\rho^{-1}.
\end{split}
\end{equation}
When $m_G\neq m_g$, the high-SNR approximation of $\mathbb{P}_F^{do}$ is given by
\begin{equation}\label{eq-OMA-OP4}
\begin{split}
\mathbb{P}_F^{do,\infty}=\frac{\tilde{m}^K\tilde{c}_2^{2m_sK}}{\Gamma(2m_sK+1)}\rho^{-m_sK},
\end{split}
\end{equation}
where $\tilde{c}_2=\sqrt{\frac{\tilde{\gamma}_F^o}{c}}$.
\end{proposition}

\begin{proof}
Similar to the proof of Proposition \ref{Prop-Downlink-OP}.
\end{proof}

\begin{corollary}\label{Coro-OMA-OP}
When N and F are under the downlink OMA scheme, the diversity orders of N and F are given by $\mathcal{D}_N^{do}=1$ and $\mathcal{D}_F^{do}=m_sK$, respectively.
\end{corollary}

\begin{remark}
The diversity order of F for the downlink OMA is affected by the number of reflecting elements and Nakagami fading parameters. Moreover, the diversity orders of N and F for the downlink OMA are the same as these for the downlink NOMA, respectively.
\end{remark}

\subsubsection{Ergodic Rate}
The ERs of N and F for OMA are presented in the following theorem.
\begin{theorem}\label{Theo-OMA-ER}
When N and F are under the downlink OMA scheme, their ERs are, respectively, given by
\begin{equation}\label{eq-OMA-ER1}
\begin{split}
R_N^{do}=-\frac{e^{\frac{1}{a \rho}}}{2\ln(2)}\mathrm{Ei}\left(-\frac{1}{a \rho}\right),
\end{split}
\end{equation}
and
\begin{equation}\label{eq-OMA-ER2}
\begin{split}
R_F^{do}\approx \frac{\lambda^{\frac{1}{4}}}{4}e^{-\frac{\lambda}{2}}
\sum_{l=1}^{u_2}\omega_{2,l}\mathcal{J}_2(x_{2,l}),
\end{split}
\end{equation}
where $u_2$ is the number of nodes for the Gauss-Laguerre quadrature and determines the approximation precision, $x_{2,l}$ is the $l$th root of Laguerre polynomial $L_{u_2}(x)$, $\omega_{2,l}=\frac{x_{2,l}}{(u_2+1)^2\left(L_{u_2+1}(x_{2,l})\right)^2}$ is the weight, and
\begin{equation}
\begin{split}
\mathcal{J}_2(x)=x^{-\frac{1}{4}}
e^{\frac{x}{2}}\log_2\left(1+b\rho x\right)I_{-\frac{1}{2}}\left(\sqrt{\lambda x}\right).
\end{split}
\end{equation}
\end{theorem}

\begin{proof}
The analysis for N is similar to that in the proof of Theorem \ref{Theo-Downlink-ER}.
For the ER of F, it can be expressed as
\begin{equation}
\begin{split}
R_F^{do}\approx \int_0^{\infty}\frac{1}{2}\log_2(1+b\rho x)f_X(x)dx
=\frac{\lambda^{\frac{1}{4}}}{4}e^{-\frac{\lambda}{2}}\underbrace{\int_{0}^{\infty}x^{-\frac{1}{4}}
e^{-\frac{x}{2}}\log_2\left(1+b\rho x\right)I_{-\frac{1}{2}}\left(\sqrt{\lambda x}\right)dx}_{J_2}.
\end{split}
\end{equation}
Next, $J_2$ can be approximated by using the Gauss-Laguerre quadrature. As such, we have
\begin{equation}
\begin{split}
J_2\simeq \sum_{l=1}^{u_2}\omega_{2,l}\mathcal{J}_2(x_{2,l}).
\end{split}
\end{equation}
This completes the proof.
\end{proof}

\begin{remark}
The ER of F is accurate as the calculation is not related to $\int_{0}^{0^{+}}g(x)f_X(x)dx$.
On the other hand, both ERs of N and F increase as the SNR increases.
\end{remark}

Then, the approximations in the high-SNR regime are given in the following proposition.
\begin{proposition}\label{Prop-OMA-ER}
In the high-SNR regime, the approximation of $R_N^{do}$ is given by
\begin{equation}\label{eq-OMA-ER3}
\begin{split}
R_N^{do,\infty}=\frac{1}{2}\left(\log_2(a\rho)-\frac{E_c}{\ln(2)}\right),
\end{split}
\end{equation}
and the tight upper bound of $R_F^{do}$ is approximated as
\begin{equation}\label{eq-OMA-ER4}
\begin{split}
R_F^{do,\infty}=\frac{1}{2}\log_2\left(b\rho (1+\lambda)\right).
\end{split}
\end{equation}
\end{proposition}

\begin{proof}
The analysis for N is similar to that in the proof of Proposition \ref{Prop-Downlink-ER}.
For F, we have $\mathbb{E}(X)=1+\lambda$ according to the property of the noncentral chi-square distribution.
Then, since $\frac{1}{2}\log_2\left(1+b\rho X\right)$ with respect to $X$ is concave, with the aid of Jensen's inequality \cite{liu2017non}, we have
\begin{equation}
\begin{split}
R_F^{do}&=\mathbb{E}\left(\frac{1}{2}\log_2\left(1+b\rho X\right)\right)\le \frac{1}{2}\log_2\left(1+b\rho\mathbb{E}(X)\right)
=\frac{1}{2}\log_2\left(1+b\rho (1+\lambda)\right).
\end{split}
\end{equation}
When $\rho\rightarrow\infty$, we can obtain \eqref{eq-OMA-ER4}.
Furthermore, simulation results show that the upper bound is tight.
This completes the proof.
\end{proof}

\begin{corollary}\label{Coro-OMA-ER}
When N and F are under the downlink OMA scheme, their high-SNR slopes are given by $\mathcal{S}_N^{do}=0.5$ and $\mathcal{S}_F^{do}=0.5$, respectively.
\end{corollary}

\begin{proof}
Since $R_F^{do,\infty}$ is the upper bound in the high-SNR regime, we can only obtain the upper bound of the high-SNR slope of F, which is $0.5$.
However, as $R_F^{do,\infty}$ is tight, which can be validated by simulations in the following, the high-SNR slope of F can achieve its upper bound.
This completes the proof.
\end{proof}

\begin{remark}
The high-SNR slope of N for the downlink OMA is half of that for the downlink NOMA. On the other hand, the high-SNR slope of F in downlink OMA networks, which is not related to the number of reflecting elements and Nakagami fading parameters, is larger than that in downlink NOMA networks.
\end{remark}

\section{Performance Analysis for Uplink Transmission}
In this section, we will investigate the performance of the uplink transmission.

\subsection{NOMA}
\subsubsection{Outage Probability}
The OPs of N and F are given by
\begin{equation}
\begin{split}
\mathbb{P}_N^u=\mathrm{P_r}\left(\mathrm{SINR}_N^u < \tilde{\gamma}_N\right),
\end{split}
\end{equation}
and
\begin{equation}
\begin{split}
\mathbb{P}_F^u=1 - \mathrm{P_r}\left(\mathrm{SINR}_N^u \ge \tilde{\gamma}_N, \mathrm{SNR}_F^u \ge \tilde{\gamma}_F\right),
\end{split}
\end{equation}
respectively.

Then, the closed-form expressions for OPs can be derived, and the results are shown in the following theorem.
\begin{theorem}\label{Theo-Uplink-OP}
In the considered IRS-aided NOMA network, the OPs of N and F for the uplink are given by
\begin{equation}\label{eq-Up-NOMA-OP1}
\begin{split}
\mathbb{P}_N^u \approx 1- e^{-\frac{\tilde{\gamma}_N}{a\rho'}-\frac{\lambda}{2}}  \sum_{i=0}^{\infty}\frac{\lambda^i}{i!2^{2i+\frac{1}{2}}\left(\frac{b\tilde{\gamma}_N}{a} +\frac{1}{2}\right)^{i+\frac{1}{2}}},
\end{split}
\end{equation}
and
\begin{equation}\label{eq-Up-NOMA-OP2}
\begin{split}
\mathbb{P}_F^u \approx 1- e^{-\frac{\tilde{\gamma}_N}{a\rho'}-\frac{\lambda}{2}}  \sum_{i=0}^{\infty}\frac{\lambda^i\Gamma\left(i+\frac{1}{2},\frac{\tilde{\gamma}_N\tilde{\gamma}_F}{a\rho'}
+\frac{\tilde{\gamma}_F}{2b\rho'}\right)}{i!2^{2i+\frac{1}{2}}\left(\frac{b\tilde{\gamma}_N}{a} +\frac{1}{2}\right)^{i+\frac{1}{2}}\Gamma\left(i+\frac{1}{2}\right)},
\end{split}
\end{equation}
respectively, where $\Gamma(\cdot,\cdot)$ is the upper incomplete gamma function.
\end{theorem}

\begin{proof}
The OP of N can be transformed into
\begin{equation}
\begin{split}
\mathbb{P}_N^u&=\mathrm{P_r}\left(\frac{aY}{bX+ 1/\rho'} < \tilde{\gamma}_N\right)
\approx \int_{0}^{\infty}\int_{0}^{\frac{b\tilde{\gamma}_N x}{a}+\frac{\tilde{\gamma}_N}{a\rho'}}f_{Y}(y)dy f_X(x)dx\\
&=1-e^{-\frac{\tilde{\gamma}_N}{a\rho'}-\frac{\lambda}{2}}
\sum_{i=0}^{\infty}\frac{\lambda^i}{i!2^{2i+\frac{1}{2}}\Gamma\left(i+\frac{1}{2}\right)}\int_{0}^{\infty} e^{-\left(\frac{b\tilde{\gamma}_N}{a} +\frac{1}{2}\right)x} x^{i-\frac{1}{2}}dx.
\end{split}
\end{equation}
Furthermore, by referring to \cite[\textrm{eq}. (3.381.4)]{gradshteyn2007}, \eqref{eq-Up-NOMA-OP1} can be derived.
Similarly, the OP of F can be first transformed into
\begin{equation}
\begin{split}
\mathbb{P}_F^u&=1-\mathrm{P_r}\left(\frac{aY}{bX+1/\rho'} \ge \tilde{\gamma}_N,b\rho' X\ge \tilde{\gamma}_F\right)
\approx 1-\int_{\frac{\tilde{\gamma}_F}{b\rho'}}^{\infty}\int_{\frac{b\tilde{\gamma}_N x}{a}+\frac{\tilde{\gamma}_N}{a\rho'}}^{\infty}f_{Y}(y)dy f_X(x)dx\\
&=1-e^{-\frac{\tilde{\gamma}_N}{a\rho'}-\frac{\lambda}{2}}  \sum_{i=0}^{\infty}\frac{\lambda^i}{i!2^{2i+\frac{1}{2}}\Gamma\left(i+\frac{1}{2}\right)}\int_{\frac{\tilde{\gamma}_F}{b\rho'}}^{\infty} e^{-\left(\frac{b\tilde{\gamma}_N}{a} +\frac{1}{2}\right)x} x^{i-\frac{1}{2}}dx.
\end{split}
\end{equation}
Then, by referring to \cite[\textrm{eq}. (3.381.3)]{gradshteyn2007}, we can derive the closed-form expression as \eqref{eq-Up-NOMA-OP2}.
This completes the proof.
\end{proof}

\begin{remark}\label{Rema-10}
The OP of N is accurate since it does not involve the integral of $\int_{0}^{0^{+}}g(x)f_X(x)dx$.
Although the calculation of F's OP is related to $\int_{0}^{0^{+}}g(x)f_X(x)dx$, it only has a small bias in the medium-SNR regime.
This is because F's OP must be higher than N's OP and converges to a floor in the high-SNR regime with N's OP due to the principle of uplink NOMA.
\end{remark}

The floor for OPs of both UDs is presented in the following proposition.
\begin{proposition}\label{Prop-Uplink-OP}
In the high-SNR regime, $\mathbb{P}_N^u$ and $\mathbb{P}_F^u$ approach the same floor as
\begin{equation}\label{eq-Up-NOMA-OP3}
\begin{split}
\mathbb{P}_{both}^{u,\infty}=1- e^{-\frac{\lambda}{2}} \sum_{i=0}^{\infty}\frac{\lambda^i}{i!2^{2i+\frac{1}{2}}\left(\frac{b\tilde{\gamma}_N}{a} +\frac{1}{2}\right)^{i+\frac{1}{2}}}.
\end{split}
\end{equation}
\end{proposition}

\begin{proof}
We have that
$\lim\limits_{\rho'\rightarrow\infty}e^{-\frac{\tilde{\gamma}_N}{a\rho'}} = 1$ and $\lim\limits_{\rho'\rightarrow\infty}\Gamma\left(i+\frac{1}{2},\frac{\tilde{\gamma}_N\tilde{\gamma}_F}{a\rho'}
+\frac{\tilde{\gamma}_F}{2b\rho'}\right)=\Gamma\left(i+\frac{1}{2}\right)$.
Thus, the limit of both $\mathbb{P}_N^u$ and $\mathbb{P}_F^u$ for $\rho'\rightarrow\infty$ can be obtained.
This completes the proof.
\end{proof}

\begin{corollary}\label{Coro-Uplink-OP}
In the considered IRS-aided NOMA network, both diversity orders of N and F for the uplink are $0$, i.e., $\mathcal{D}_N^u=0$ and $\mathcal{D}_F^u=0$.
\end{corollary}

\begin{remark}
The diversity order of F for the uplink NOMA is not related to the number of reflecting elements and Nakagami fading parameters.
\end{remark}

\subsubsection{Ergodic Rate}
The ERs of the N and F are given by
\begin{equation}\label{eq-UP-NOMA-ERN}
\begin{split}
R_N^u=\mathbb{E}\left(\log_2\left(1+\mathrm{SINR}_N^u\right)\right),
\end{split}
\end{equation}
and
\begin{equation}\label{eq-UP-NOMA-ERF}
\begin{split}
R_F^u=\mathbb{E}\left(\log_2\left(1+\mathrm{SNR}_F^u\right)\right),
\end{split}
\end{equation}
respectively.

Then, the ERs of N and F can be obtained, and the results are provided in the following theorem.
\begin{theorem}\label{Theo-Uplink-ER}
In the considered IRS-aided NOMA network, the ER of N for the uplink is given by
\begin{equation}\label{eq-Up-NOMA-ER1}
\begin{split}
R_N^u\approx -\frac{\lambda^{\frac{1}{4}}}{2\ln(2)}e^{\frac{1}{a\rho'}-\frac{\lambda}{2}} \sum_{l=1}^{u_3}\omega_{3,l}\mathcal{J}_3(x_{3,l}),
\end{split}
\end{equation}
where $u_3$ is the number of nodes for the Gauss-Laguerre quadrature and determines the approximation precision, $x_{3,l}$ is the $l$th root of Laguerre polynomial $L_{u_3}(x)$, $\omega_{3,l}=\frac{x_{3,l}}{(u_3+1)^2\left(L_{u_3+1}(x_{3,l})\right)^2}$ is the weight, and
\begin{equation}\label{eq-5-1}
\begin{split}
\mathcal{J}_3(x)=x^{-\frac{1}{4}}e^{\left(\frac{b}{a}+\frac{1}{2}\right)x} \mathrm{Ei}\left(-\frac{b}{a}x-\frac{1}{a\rho'}\right)
I_{-\frac{1}{2}}\left(\sqrt{\lambda x}\right).
\end{split}
\end{equation}
On the other hand, the ER of F for the uplink is given by
\begin{equation}\label{eq-Up-NOMA-ER2}
\begin{split}
R_F^u \approx \frac{\lambda^{\frac{1}{4}}}{2}e^{-\frac{\lambda}{2}}
\sum_{l=1}^{u_4}\omega_{4,l}\mathcal{J}_4(x_{4,l}),
\end{split}
\end{equation}
where $u_4$ is the number of nodes for the Gauss-Laguerre quadrature and determines the approximation precision, $x_{4,l}$ is the $l$th root of Laguerre polynomial $L_{u_4}(x)$, $\omega_{4,l}=\frac{x_{4,l}}{(u_4+1)^2\left(L_{u_4+1}(x_{4,l})\right)^2}$ is the weight, and
\begin{equation}
\begin{split}
\mathcal{J}_4(x)=x^{-\frac{1}{4}}
e^{\frac{x}{2}}\log_2\left(1+b\rho' x\right)I_{-\frac{1}{2}}\left(\sqrt{\lambda x}\right).
\end{split}
\end{equation}
\end{theorem}

\begin{proof}
See Appendix \ref{Appen-Theo-Uplink-ER}.
\end{proof}

\begin{remark}
Both ERs are accurate as they are not related to the integral of $\int_{0}^{0^{+}}g(x)f_X(x)dx$.
Due to the principle of uplink NOMA, the ER of N approaches a ceiling when $\rho'\rightarrow\infty$, while the ER of F increases as the SNR increases.
\end{remark}

The ceiling for  N and the high-SNR approximation for F are presented in the following proposition.
\begin{proposition}\label{Prop-Uplink-ER}
In the high-SNR regime, $R_N^u$ approaches a ceiling that is given by
\begin{equation}\label{eq-Up-NOMA-ER3}
\begin{split}
R_N^{u,\infty} \approx -\frac{\lambda^{\frac{1}{4}}}{2\ln(2)}e^{-\frac{\lambda}{2}} \sum_{l=1}^{u_3}\omega_{3,l}\mathcal{J}_3^{\infty}\left(x_{3,l}\right),
\end{split}
\end{equation}
where
\begin{equation}
\begin{split}
\mathcal{J}_3^{\infty}(x)=x^{-\frac{1}{4}}e^{\left(\frac{b}{a}+\frac{1}{2}\right)x} \mathrm{Ei}\left(-\frac{b}{a}x\right)
I_{-\frac{1}{2}}\left(\sqrt{\lambda x}\right).
\end{split}
\end{equation}
On the other hand, the tight upper bound of $R_F^u$'s high-SNR approximation is given by
\begin{equation}\label{eq-Up-NOMA-ER4}
\begin{split}
R_F^{u,\infty}=\log_2\left(b\rho' (1+\lambda)\right).
\end{split}
\end{equation}
\end{proposition}

\begin{proof}
For N, when $\rho'\rightarrow\infty$, we have $\frac{1}{a\rho'}\rightarrow0$. Hence, based on \eqref{eq-Up-NOMA-ER1} and \eqref{eq-5-1}, we can obtain \eqref{eq-Up-NOMA-ER3}.
The analysis for F is similar to that in the proof of Proposition \ref{Prop-OMA-ER}.
This completes the proof.
\end{proof}

\begin{corollary}\label{Coro-Uplink-ER}
In the considered IRS-aided NOMA network, the high-SNR slopes of N and F for the uplink are given by $\mathcal{S}_N^u=0$ and $\mathcal{S}_F^u=1$, respectively.
\end{corollary}

\begin{proof}
Similar to the proof of Corollary \ref{Coro-OMA-ER}.
\end{proof}

\begin{remark}
The high-SNR slope of F for the uplink NOMA is not affected by the number of reflecting elements and Nakagami fading parameters.
\end{remark}

\subsection{OMA}
For OMA, the performance analyses for the uplink and the downlink are similar, and we can obtain all results for the uplink by referring to Subsection \ref{Subsec-OMA} in a simple method, i.e., replacing the transmit SNR of the BS $(\rho)$ with the transmit SNR of the UD $(\rho')$.

\subsection{Summary of All Results}
After completing all analyses for downlink and uplink networks, all results related to diversity order and high-SNR slope are summarized in Table \ref{Tab_Results} for ease of reference.
\begin{table}
\centering
\caption{Diversity order ($\mathcal{D}$) and high-SNR slope ($\mathcal{S}$) for each scenario}
\begin{tabular}{|c|c|c|c|c|c|}
\hline
\centering
\multirow{2}{*}{Multiple-access scheme} & \multirow{2}{*}{UD} & \multicolumn{2}{c|}{Downlink} & \multicolumn{2}{c|}{Uplink}  \\
\cline{3-6}                             & & $\mathcal{D}$ & $\mathcal{S}$ & $\mathcal{D}$ & $\mathcal{S}$ \\
\hline

\multirow{2}{*}{NOMA}                   & N & $1$ & $1$ & $0$ & $0$  \\
\cline{2-6}                             & F & $m_sK$ & $0$ & $0$ & $1$ \\
\hline

\multirow{2}{*}{OMA}                    & N & $1$ & $0.5$ & $1$ & $0.5$  \\
\cline{2-6}                             & F & $m_sK$ & $0.5$ & $m_sK$ & $0.5$ \\
\hline
\end{tabular}\label{Tab_Results}
\end{table}

\begin{remark}
We observe that OMA has higher diversity orders than NOMA in uplink networks. It only reflects that OMA outperforms NOMA in the high-SNR regime for uplink networks with the fixed-rate transmission. In the low-SNR regime, NOMA has better performance than OMA \cite{wei2019performance}.
\end{remark}

\section{Numerical Results and Discussion}

\begin{table*}
\centering
\caption{Parameters setting}
\begin{tabular}{|c|c|}
\hline
Bandwidth & $B=1$ MHz\\
\hline
Amplitude-reflection coefficient of the IRS & $\beta=0.9$\\
\hline
Distances & $d_N=10$ m, $d_{F1}=40$ m, and $d_{F2}=10$ m\\
\hline
Path-loss exponents & $\alpha_h=3.5$, $\alpha_G=2.5$, and $\alpha_g=2.5$\\
\hline
Nakagami fading parameters& $m_G=3$ and $m_g=1.5$\\
\hline
Target data-rates for the fixed-rate transmission & $\tilde{R}_N=\tilde{R}_F=0.1$ Mbps\\
\hline
Number of points for Chebyshev-Gauss and Gauss-Laguerre quadratures& $u_1=u_2=u_3=u_4=100$\\
\hline
\end{tabular}\label{Tab_Parameters}
\end{table*}

In this section, numerical results are presented for the performance evaluation of the considered network. Meanwhile, Monte Carlo simulations are conducted to verify the accuracy. The parameters are set as shown in Table \ref{Tab_Parameters} \cite{ding2020simple,ding2020impact,hou2019reconfigurable}.

For comparisons, we regard an FDR-aided NOMA network as the benchmark. Specifically, an FDR under the classic protocol is deployed at the place of the IRS to help F to communicate with the BS.
The FDR works under a realistic assumption that is the same as \cite{zhong2016non}. Specifically, since the relay is aware of its own transmitted symbol, self-interference cancellation can be applied. we assume that the self-interference channel experiences the Nakagami-$m$ fading and the self-interference cancellation is imperfect.
Since the reflection at the IRS is passive without consuming the energy, for fairness, we assume that the transmit power at the BS and the FDR is $P_b^r=0.5P_b$ for the downlink, and the transmit power at F and the FDR is $P_u^r=0.5P_u$ for the uplink.

\subsection{Downlink Networks}

\begin{figure}
\centering
     \begin{subfigure}[OPs derived from the CLT-based channel statistics when $K=8$, $\alpha_1=0.1$, and $\alpha_2=0.9$.]{
         \label{fig_OP_Downlink}
         \includegraphics[width=0.55\textwidth]{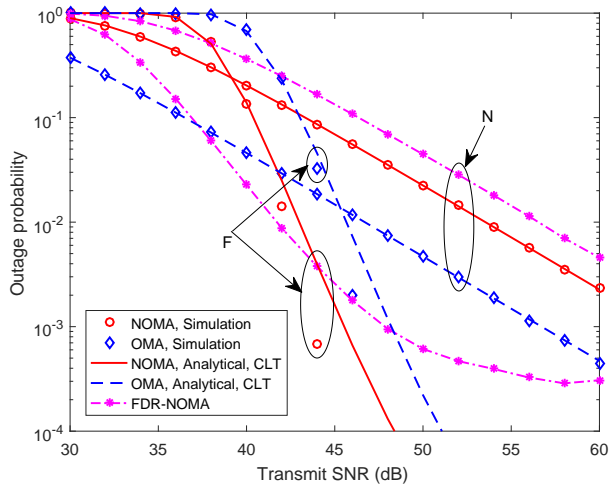}}
     \end{subfigure}
    \hfill
     \begin{subfigure}[High-SNR approximations of OPs derived from the LT-based channel statistics when $K=2$.]{
         \label{fig_OP_Downlink_High_SNR}
         \includegraphics[width=0.55\textwidth]{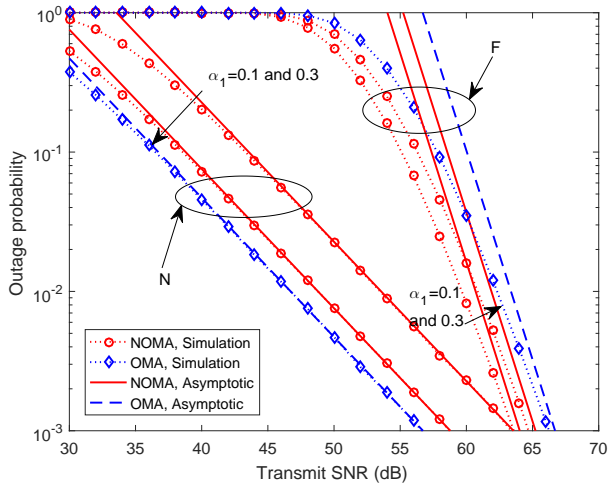}}
     \end{subfigure}
     \caption{OPs versus the transmit SNR in downlink networks.}
\end{figure}

In Fig. \ref{fig_OP_Downlink}, the OPs versus the transmit SNR in IRS-aided NOMA, IRS-aided OMA, and FDR-aided NOMA networks of downlink are plotted.
First, the analyses of N for both NOMA and OMA are accurate as all simulation results coincide with the corresponding analytical results that are derived from \eqref{eq-NOMA-OP1} and \eqref{eq-OMA-OP1}.
For F, the analyses under these two schemes that are derived from \eqref{eq-NOMA-OP2} and \eqref{eq-OMA-OP2} are accurate in the low-SNR regime but are inaccurate in the high-SNR regime, which results from the use of the CLT-based channel statistics.
As a benchmark, the OP curves for the FDR-aided system are plotted for comparisons. We observe that N in the IRS-aided NOMA system always has better performance than that in the FDR-aided NOMA system, since the transmit power of the BS in the former system is twice that in the latter.
For F, the FDR-aided system has better performance than the IRS-aided system in the low-SNR regime, since the IRS transmission experiences severe path loss \cite{ding2020impact}. Nevertheless, in the high-SNR regime, the IRS-aided system has much better performance. One reason is that the OP of F in the FDR-aided system converges to a floor in the high-SNR regime due to the residual self-interference. On the other hand, F in the IRS-aided system has a large diversity order, which is the advantage of the IRS over FDR.

Since the theoretical results of F are not accurate in the high-SNR regime in Fig. \ref{fig_OP_Downlink}, we further plot the high-SNR approximation curves when $\alpha_1=0.1$ ($\alpha_2=0.9$) and $\alpha_1=0.3$ ($\alpha_2=0.7$) in Fig. \ref{fig_OP_Downlink_High_SNR}.
It is observed that the OPs of N and F for NOMA and OMA gradually approach their respective asymptotic curves derived from \eqref{eq-NOMA-OP3}, \eqref{eq-NOMA-OP4}, \eqref{eq-OMA-OP3}, and \eqref{eq-OMA-OP4}, which validates our analysis.
We also observe that a larger $\alpha_1$ ($\alpha_2=1-\alpha_1$) incurs a lower OP of N and a higher OP of F.
Furthermore, by observing slopes, we find that both diversity orders of N under NOMA and OMA schemes are $1$ when $K=2$.
For F, the diversity orders under NOMA and OMA schemes are the same and equal $3$ for $K=2$, which is consistent with Corollary \ref{Coro-Downlink-OP} and Corollary \ref{Coro-OMA-OP}.
Therefore, the fixed power-allocation coefficients affect the OP but have no effect on the diversity order.

\begin{figure}
\centering
     \begin{subfigure}[ERs derived from the CLT-based channel statistics when $\alpha_1=0.1$ and $\alpha_2=0.9$.]{
         \label{fig_ER_Downlink}
         \includegraphics[width=0.55\textwidth]{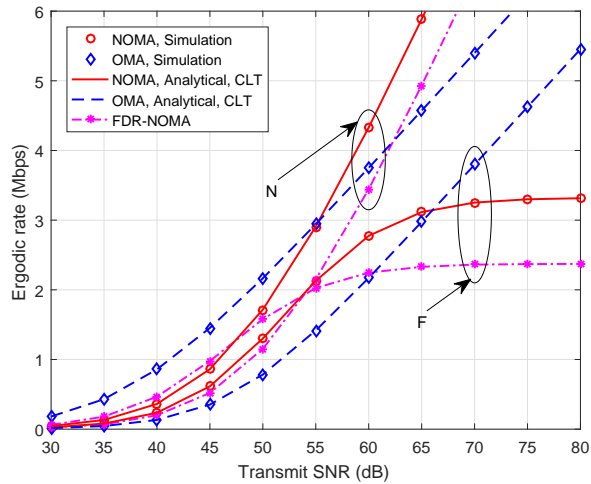}}
     \end{subfigure}
    \hfill
     \begin{subfigure}[High-SNR approximations of ERs.]{
         \label{fig_ER_Downlink_High_SNR}
         \includegraphics[width=0.55\textwidth]{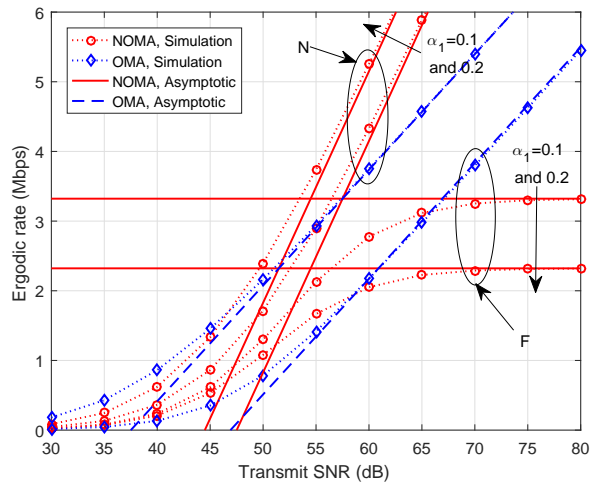}}
     \end{subfigure}
     \caption{ERs versus the transmit SNR in downlink networks when $K=10$.}
\end{figure}

In Fig. \ref{fig_ER_Downlink}, the ER curves for different downlink networks are depicted.
First, it is observed that the simulation points for IRS-aided NOMA and OMA networks match well with their respective analytical results derived from \eqref{eq-NOMA-ER1}, \eqref{eq-NOMA-ER2}, \eqref{eq-OMA-ER1}, and \eqref{eq-OMA-ER2}.
As a benchmark, the ER curves of the FDR-aided NOMA are plotted for comparisons. It is observed that N in the IRS-aided system always has a higher ER than that in the FDR-aided system, since the transmit power of the BS in the former system is twice that in the latter.
On the other hand, for F, the ER in the IRS-aided system is lower than that in the FDR-aided system in the low-SNR regime. This is because the IRS transmission experiences severe path loss \cite{ding2020impact}. As the SNR increases, the ER of F in the IRS-aided system approaches a higher ceiling than that in the FDR-aided system due to the existence of the residual self-interference when using the FDR, which demonstrates the superiority of IRS again.

Furthermore, the asymptotic analysis is validated when $\alpha_1=0.1$ ($\alpha_2=0.9$) and $\alpha_1=0.2$ ($\alpha_2=0.8$) in Fig. \ref{fig_ER_Downlink_High_SNR}.
It is observed that the high-SNR approximations that are derived from \eqref{eq-NOMA-ER3}, \eqref{eq-NOMA-ER4}, \eqref{eq-OMA-ER3}, and \eqref{eq-OMA-ER4} are accurate.
We also observe that a larger $\alpha_1$ ($\alpha_2=1-\alpha_1$) incurs a higher ER of N and a lower ER of F.
Also, we observe that the high-SNR slope of N for NOMA is $1$,\footnote{The high-SNR slope of $1$ corresponds to the slope of $0.1\log_210\approx 0.33$ in Fig. \ref{fig_ER_Downlink_High_SNR}.} while that for OMA is $0.5$.
Due to the principle of downlink NOMA, the ER of F approaches a ceiling in the high-SNR regime, i.e., the high-SNR slope of F for NOMA is $0$.
Meanwhile, F has a high-SNR slope of $0.5$ under the OMA scheme.
These observations coincide with Corollary \ref{Coro-Downlink-ER} and Corollary \ref{Coro-OMA-ER}.
Therefore, the fixed power-allocation coefficients affect the ER but have no effect on the high-SNR slope.

\subsection{Uplink Networks}
\begin{figure}
\begin{center}
\includegraphics[width=0.55\textwidth]{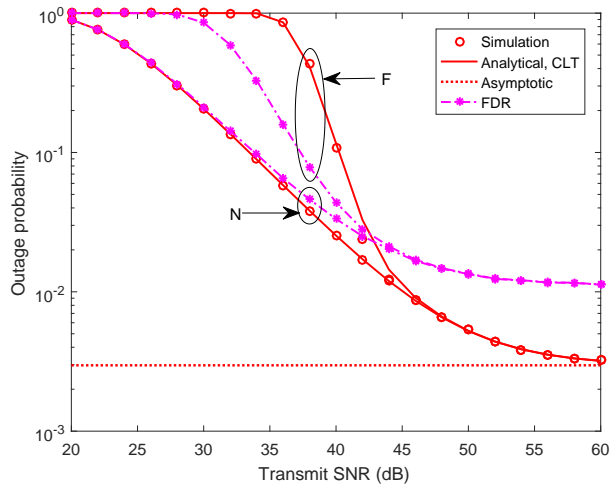}
\end{center}
\caption{OPs versus the transmit SNR in uplink NOMA networks when $K=8$.}
\label{fig_OP_Uplink}
\end{figure}

In Fig. \ref{fig_OP_Uplink}, the OPs versus the transmit SNR in uplink NOMA networks are plotted.
For the IRS-aided NOMA network, it is observed that the simulation points of N match well with the analytical results derived from \eqref{eq-Up-NOMA-OP1}.
The simulation results of F have a small bias with the analytical results derived from \eqref{eq-Up-NOMA-OP2} in the medium-SNR regime (e.g., 42 dB), which is consistent with our explanation in Remark \ref{Rema-10}.
Then, we observe that both OPs of N and F decrease as the increase of the transmit SNR and gradually approach a floor that is derived from \eqref{eq-Up-NOMA-OP3}.
This is because of the principle of uplink NOMA that F's signal is regarded as interference to decode N's signal.
Thus, the diversity orders of both UDs are $0$, which is consistent with Corollary \ref{Coro-Uplink-OP}.
Finally, it is observed that the OPs for the FDR-aided NOMA network also approach a floor.
In Fig. \ref{fig_OP_Uplink}, the floor for the IRS-aided network is lower than that for the FDR-aided network.
However, the former will be higher than the latter if we increase $K$.
Meanwhile, the floor for the FDR-aided network is also determined by the efficiency of the self-interference cancellation.
Thus, for uplink IRS-aided networks, the contribution of the IRS to OP is not obvious.

\begin{figure}
\begin{center}
\includegraphics[width=0.55\textwidth]{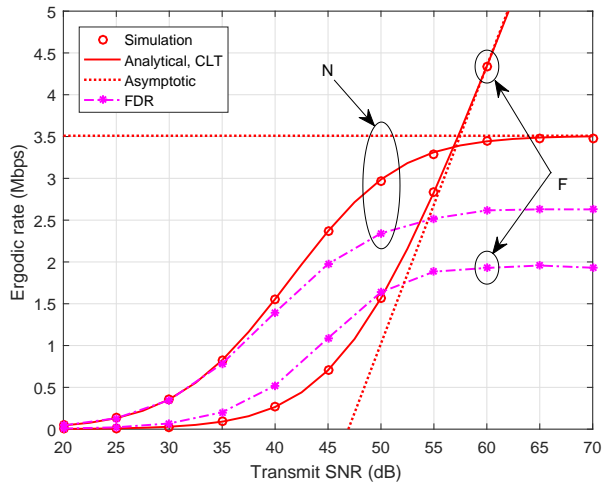}
\end{center}
\caption{ERs versus the transmit SNR in uplink NOMA networks when $K=10$.}
\label{fig_ER_Uplink}
\end{figure}

In Fig. \ref{fig_ER_Uplink}, the ER curves for different uplink networks are depicted.
For the IRS-aided NOMA network, we observe that the simulated results coincide with the corresponding analytical results that are derived from \eqref{eq-Up-NOMA-ER1} and \eqref{eq-Up-NOMA-ER2}.
In addition, the approximations in the high-SNR regime that are derived from \eqref{eq-Up-NOMA-ER3} and \eqref{eq-Up-NOMA-ER4} are also asymptotic exact.
Following that, it is observed that the high-SNR slopes of N and F are $0$ and $1$, respectively, which is consistent with Corollary \ref{Coro-Uplink-ER}.
Lastly, we observe that the ER of N in the FDR-aided NOMA network also converges to a ceiling.
For the ER of F, it remains increasing with the increase of the transmit SNR in the IRS-aided network, while it has a ceiling in the FDR-aided network due to the residual self-interference, which reveals the advantage of the IRS.

\section{Conclusion}

In this paper, we have characterized the system performance of IRS-aided NOMA and OMA networks for downlink and uplink transmissions.
We have demonstrated that the use of IRS can significantly enhance system performance, especially improving the diversity order.
We have assumed the fixed power allocation for downlink NOMA, which affects the OP and ER but has no influence on the diversity order and high-SNR slope.
Furthermore, simulation results have demonstrated that IRS outperforms FDR in the high-SNR regime.
Since this work focuses on SISO networks, IRS-aided MIMO networks are worthy of investigation for future work.
Moreover, optimizing power allocation for NOMA can further improve system performance, which is also a direction of future research.

\begin{appendices}

\section{Proof of Lemma \ref{Lemm-Channel-CLT}}\label{Appen-Lemm-Channel-CLT}
\renewcommand{\theequation}{\thesection.\arabic{equation}}
\setcounter{equation}{0}

Based on the property of the Nakagami-$m$ fading model, the expectation and variance of $|\mathcal{G}_{k}|$ are $\mu_\mathcal{G}=\left(\frac{1}{m_\mathcal{G}}\right)^\frac{1}{2}\frac{\Gamma\left(m_\mathcal{G}+\frac{1}{2}\right)}{\Gamma(m_\mathcal{G})}$ and $\mathrm{Var}_\mathcal{G}=1-\frac{1}{m_\mathcal{G}}\left(\frac{\Gamma\left(m_\mathcal{G}+\frac{1}{2}\right)}{\Gamma(m_\mathcal{G})}\right)^2$, respectively, where $\mathcal{G} \in \{G, g\}$.
Then, the expectation and variance of $|G_{k}||g_{k}|$ are $\mu_p=\mu_g\mu_G=\sqrt{\xi}$ and $\mathrm{Var}_p=\big(\mathrm{Var}_g+\mu_g^2\big)\big(\mathrm{Var}_G+\mu_G^2\big)-\mu_g^2\mu_G^2=1-\xi$, respectively. Next, since all $|G_{k}||g_{k}|$ $(k = 1, 2, \cdots, K)$ are independent and identically distributed (i.i.d.), based on the CLT, $\sum_{k=1}^{K}|G_{k}||g_{k}|$ tends towards a Gaussian distribution as
\begin{equation}
\begin{split}
\sum_{k=1}^{K}|G_{k}||g_{k}| \dot\thicksim \mathcal{N}\left(K\mu_p,K\mathrm{Var}_p\right).
\end{split}
\end{equation}
Furthermore, we unify the variance and have
\begin{equation}
\begin{split}
\frac{\sum_{k=1}^{K}|G_{k}||g_{k}|}{\sqrt{K\mathrm{Var}_p}} \dot\thicksim \mathcal{N}\left(\frac{\sqrt{K}\mu_p}{\sqrt{\mathrm{Var}_p}},1\right).
\end{split}
\end{equation}
Hence, $X=\frac{\left(\sum_{k=1}^{K}|G_{k}||g_{k}|\right)^2}{K\mathrm{Var}_p}$ follows a noncentral chi-square distribution  $\chi_1^{'2}(\lambda)$ with a PDF given by
\begin{equation}
\begin{split}
f_X(x)=\frac{\lambda^{\frac{1}{4}}}{2}e^{-\frac{x+\lambda}{2}}x^{-\frac{1}{4}}I_{-\frac{1}{2}}\left(\sqrt{\lambda x}\right),
\end{split}
\end{equation}
where $\lambda=\frac{K\mu_p^2}{\mathrm{Var}_p}$ and
\begin{equation}
\begin{split}
I_{-\frac{1}{2}}\left(\sqrt{\lambda x}\right)=\left(\frac{\sqrt{\lambda x}}{2}\right)^{-\frac{1}{2}}\sum_{i=0}^{\infty}\frac{\left(\frac{\lambda x}{4}\right)^i}{i!\Gamma(i+\frac{1}{2})}.
\end{split}
\end{equation}
The CDF of $X$ is given by
\begin{equation}
\begin{split}
F_X(x)=1-Q_{\frac{1}{2}}\left(\sqrt{\lambda},\sqrt{x}\right),
\end{split}
\end{equation}
where
\begin{equation}
\begin{split}
Q_{\frac{1}{2}}\left(\sqrt{\lambda},\sqrt{x}\right)=1-e^{-\frac{\lambda}{2}}
\sum_{i=0}^{\infty}\left(\frac{\lambda}{2}\right)^i\frac{\gamma\left(i+\frac{1}{2},\frac{x}{2}\right)}{i!\Gamma\left(i+\frac{1}{2}\right)}.
\end{split}
\end{equation}
This completes the proof.

\section{Proof of Lemma \ref{Lemm-Channel-Appro}}\label{Appen-Lemm-Channel-Appro}
\renewcommand{\theequation}{\thesection.\arabic{equation}}
\setcounter{equation}{0}

Denote that $Q_k=|G_{k}||g_{k}|$.
The PDF of $Q_k$ (the product of two Nakagami-$m$ random variables) \cite{bhargav2018product} is given by
\begin{equation}
\begin{split}
f_{Q_k}(q)=\frac{4(m_sm_l)^{\frac{m_s+m_l}{2}}}{\Gamma(m_s)\Gamma(m_l)}q^{m_s+m_l-1}K_{m_s-m_l}(2\sqrt{m_sm_l}q),
\end{split}
\end{equation}
for $q\ge0$, where $K_{v}(\cdot)$ is the modified Bessel function of the second kind.
The LT of $f_{Q_k}$ is derived as
\begin{equation}
\begin{split}
\mathcal{L}_{f_{Q_k}}(s)=\frac{4(m_sm_l)^{\frac{m_s+m_l}{2}}}{\Gamma(m_s)\Gamma(m_l)}
\int_0^{\infty}q^{m_s+m_l-1}e^{-sq}
K_{m_s-m_l}(2\sqrt{m_sm_l}q)dq.
\end{split}
\end{equation}
Furthermore, by referring to \cite[\textrm{eq}. (6.621.3)]{gradshteyn2007}, we have
\begin{equation}\label{eq-Appe-3}
\begin{split}
\mathcal{L}_{f_{Q_k}}(s)=\phi\left(s+2\sqrt{m_sm_l}\right)^{-2m_s} F\left(2m_s,m_s-m_l+\frac{1}{2};m_s+m_l+\frac{1}{2};\frac{s-2\sqrt{m_sm_l}}{s+2\sqrt{m_sm_l}}\right),
\end{split}
\end{equation}
where $\phi=\frac{\sqrt{\pi}4^{m_s-m_l+1}(m_sm_l)^{m_s}\Gamma(2m_s)\Gamma(2m_l)}{\Gamma(m_s)\Gamma(m_l)\Gamma\left(m_s+m_l+\frac{1}{2}\right)}$ and $F(\cdot,\cdot;\cdot;\cdot)$ is the hypergeometric series.
We observe that \eqref{eq-Appe-3} is so complicated that it is impossible to perform the inverse LT for $\prod_{k=1}^{K}\mathcal{L}_{f_{Q_k}}(s)$. To address it, we simplify \eqref{eq-Appe-3} by assuming $s\rightarrow\infty$ to obtain the PDF for the channel gain near 0.
When $s\rightarrow\infty$, since $m_s<m_l$ satisfies the condition of \cite[\textrm{eq}. (9.122.1)]{gradshteyn2007}, we have
\begin{equation}
\begin{split}
\mathcal{L}_{f_{Q_k}}^{\infty}(s)=\tilde{m}\left(s+2\sqrt{m_sm_l}\right)^{-2m_s}.
\end{split}
\end{equation}
Since all $Q_k$ $(k=1, 2, \cdots, K)$ are i.i.d., the LT of the PDF of $Z=\sum_{k=1}^{K}Q_k$ for $s\rightarrow\infty$ is given by
\begin{equation}\label{eq-Appe-1}
\begin{split}
\mathcal{L}_{f_Z}^{\infty}(s)=\prod_{k=1}^{K}\mathcal{L}_{f_{Q_k}}^{\infty}(s)=\tilde{m}^K\left(s+2\sqrt{m_sm_l}\right)^{-2m_sK}.
\end{split}
\end{equation}
Thus, by conducting the inverse LT for \eqref{eq-Appe-1}, the PDF of $Z$ for $z\rightarrow0^+$ can be derived as \eqref{eq-Lemm-1} based on \cite[\textrm{eq}. (17.13.3)]{gradshteyn2007}.
Following that, according to \cite[\textrm{eq}. (3.351.1)]{gradshteyn2007}, the CDF of $Z$ for $z\rightarrow0^+$ can be derived as \eqref{eq-Lemm-2}.
This completes the proof.

\section{Proof of Theorem \ref{Theo-Downlink-ER}}\label{Appen-Theo-Downlink-ER}
\renewcommand{\theequation}{\thesection.\arabic{equation}}
\setcounter{equation}{0}

First, we have
\begin{equation}
\begin{split}
R_N^d=-\int_{0}^{\infty}\log_2\left(1+a\alpha_1 \rho y\right)d\left(1-F_Y(y)\right)
=\frac{1}{\ln(2)}\int_{0}^{\infty}\frac{e^{-y}}{y+\frac{1}{a\alpha_1 \rho}} dy.
\end{split}
\end{equation}
Then, by referring to \cite[\textrm{eq}. (3.352.4)]{gradshteyn2007}, we can derive \eqref{eq-NOMA-ER1}.
To obtain $R_F^d$, we denote that $\tilde{X}=\frac{b\alpha_2 X}{b\alpha_1 X + 1/\rho}$, and the corresponding CDF is given by
\begin{equation}
\begin{split}
F_{\tilde{X}}(\tilde{x})=F_X\left(\frac{\tilde{x}}{b\rho(\alpha_2-\tilde{x}\alpha_1)}\right)
=e^{-\frac{\lambda}{2}}\sum_{i=0}^{\infty}\frac{\lambda^i\gamma\left(i+\frac{1}{2},\frac{\tilde{x}}{2b\rho (\alpha_2-\tilde{x}\alpha_1)}\right)}{i!2^i\Gamma\left(i+\frac{1}{2}\right)}.
\end{split}
\end{equation}
Hence, we have
\begin{equation}
\begin{split}
R_F^d&\approx -\int_{0}^{\tilde{\alpha}}\log_2\left(1+\tilde{x}\right)d\left(1-F_{\tilde{X}}(\tilde{x})\right)
=\frac{1}{\ln(2)}\int_{0}^{\tilde{\alpha}}\frac{1-F_{\tilde{X}}(\tilde{x})}{1+\tilde{x}}d\tilde{x}\\
&=\log_2(1+\tilde{\alpha})-\frac{1}{\ln(2)}e^{-\frac{\lambda}{2}}
\sum_{i=0}^{\infty}\frac{\lambda^i}{i!2^i\Gamma(i+\frac{1}{2})}
\underbrace{\int_{0}^{\tilde{\alpha}}\frac{\gamma\left(i+\frac{1}{2},\frac{\tilde{x}}{2b\rho (\alpha_2-\tilde{x}\alpha_1)}\right)}{1+\tilde{x}}d\tilde{x}}_{J_1},
\end{split}
\end{equation}
where $\tilde{\alpha}=\frac{\alpha_2}{\alpha_1}$.
To further transform $J_1$, we denote that $t=\frac{2\tilde{x}}{\tilde{\alpha}}-1$.
Then, we have
\begin{equation}
\begin{split}
J_1=\int_{-1}^{1}\frac{\gamma\left(i+\frac{1}{2},\frac{\tilde{\alpha}(1+t)}{2b\rho (2\alpha_2-\alpha_1\tilde{\alpha}(1+t))}\right)}{1+\frac{2}{\tilde{\alpha}}+t}dt.
\end{split}
\end{equation}
Next, by using the Chebyshev-Gauss quadrature, we can approximate $J_1$ as
\begin{equation}
\begin{split}
J_1\simeq \sum_{l=1}^{u_1}\omega_{1,l}\mathcal{J}_1(t_l).
\end{split}
\end{equation}
This completes the proof.

\section{Proof of Theorem \ref{Theo-Uplink-ER}}\label{Appen-Theo-Uplink-ER}
\renewcommand{\theequation}{\thesection.\arabic{equation}}
\setcounter{equation}{0}

First, the ER of N can be transformed into
\begin{equation}\label{eq-Appe-2}
\begin{split}
R_N^u&\approx -\int_{0}^{\infty}\int_{0}^{\infty}\log_2\left(1+\frac{ay}{bx+\frac{1}{\rho}}\right)d(1-F_Y(y))f_X(x)dx\\
&=\frac{1}{\ln(2)}\int_{0}^{\infty}\int_{0}^{\infty}\frac{e^{-y}}{y+\frac{b}{a}x+\frac{1}{a\rho}}dyf_X(x)dx.
\end{split}
\end{equation}
Then, by referring to \cite[\textrm{eq}. (3.352.4)]{gradshteyn2007}, \eqref{eq-Appe-2} can be rewritten as
\begin{equation}
\begin{split}
R_N^u & \approx -\frac{1}{\ln(2)}\int_{0}^{\infty}e^{\frac{b}{a}x+\frac{1}{a\rho}} \mathrm{Ei}\left(-\frac{b}{a}x-\frac{1}{a\rho}\right)f_X(x)dx\\
&=-\frac{\lambda^{\frac{1}{4}}}{2\ln(2)} e^{\frac{1}{a\rho}-\frac{\lambda}{2}}
\underbrace{\int_{0}^{\infty}x^{-\frac{1}{4}}e^{\left(\frac{b}{a}-\frac{1}{2}\right)x} \mathrm{Ei}\left(-\frac{b}{a}x-\frac{1}{a\rho}\right)
I_{-\frac{1}{2}}\big(\sqrt{\lambda x}\big)dx}_{J_3}.
\end{split}
\end{equation}
Next, by using the Gauss-Laguerre quadrature, we have
\begin{equation}
\begin{split}
J_3\simeq \sum_{l=1}^{u_3}\omega_{3,l}\mathcal{J}_3(x_{3,l}).
\end{split}
\end{equation}
For the ER of F, it can be expressed as
\begin{equation}
\begin{split}
R_F^u & \approx \int_0^{\infty} \log_2(1+b\rho x)f_X(x)dx
=\frac{\lambda^{\frac{1}{4}}}{2}e^{-\frac{\lambda}{2}}\underbrace{\int_{0}^{\infty}x^{-\frac{1}{4}}e^{-\frac{x}{2}}\log_2\left(1+b\rho x\right)I_{-\frac{1}{2}}\left(\sqrt{\lambda x}\right)dx}_{J_4}.
\end{split}
\end{equation}
Next, $J_4$ also can be approximated by adopting the Gauss-Laguerre quadrature. As such, we have
\begin{equation}
\begin{split}
J_4\simeq \sum_{l=1}^{u_4}\omega_{4,l}\mathcal{J}_4(x_{4,l}).
\end{split}
\end{equation}
This completes the proof.

\end{appendices}

\bibliographystyle{IEEEtran}
\bibliography{ref}

\end{document}